\documentclass[12pt,a4paper]{article}
\usepackage[margin=1in]{geometry}

\usepackage{float}

\usepackage{subcaption}
\usepackage{hyperref}
\usepackage{authblk}
\usepackage[english]{babel}
\usepackage[utf8x]{inputenc}
\usepackage{amsmath}
\usepackage{amsthm}
\usepackage{amssymb}
\usepackage{algorithm}
\usepackage{algorithmicx,algpseudocode}

\usepackage{graphicx}
\graphicspath{ {Images/} }

\usepackage[nottoc]{tocbibind}

\newtheorem{theorem}{Theorem}[section]
\newtheorem{corollary}{Corollary}[theorem]
\newtheorem{lemma}[theorem]{Lemma}
\newtheorem{defn}[theorem]{Definition}

\providecommand{\keywords}[1]{\textbf{\textit{Index terms---}} #1}

\usepackage{mathtools}
\DeclarePairedDelimiter\ceil{\lceil}{\rceil}
\DeclarePairedDelimiter\floor{\lfloor}{\rfloor}

\usepackage{multirow}

\title{Multidimensional segment trees can do range updates in poly-logarithmic time}

\author[1]{Nabil Ibtehaz}
\author[1]{M. Kaykobad}
\author[1,*]{M. Sohel Rahman}

\affil[ ]{ 1017052037@grad.cse.buet.ac.bd \\ \{kaykobad,msrahman\}@cse.buet.ac.bd}
\affil[1]{Department of CSE, BUET,\protect\\ ECE Building, West Palasi, Dhaka-1205, Bangladesh}
\affil[*]{Corresponding author}

\begin{document}

\clearpage
\setcounter{page}{1}

\maketitle

\begin{abstract}
Updating and querying on a range is a classical algorithmic problem with a multitude of applications. The Segment Tree data structure is particularly notable in handling the range query and update operations. A Segment Tree divides the range into disjoint segments and merges them together to perform range queries and range updates elegantly. Although this data structure is remarkably potent for 1-dimensional problems, it falls short in higher dimensions. Lazy Propagation enables the operations to be computed in $O(logn)$ time in a single dimension. However, the concept of lazy propagation could not be translated to higher-dimensional cases, which imposes a time complexity of $O(n^{k-1} \; logn)$ for operations on $k$-dimensional data. In this work, we have made an attempt to emulate the idea of lazy propagation differently so that it can be applied for 2-dimensional cases. Moreover, the proposed modification should be capable of performing most general aggregate functions similar to the original Segment Tree, and can also be extended to even higher dimensions. Our proposed algorithm manages to perform range sum queries and updates in $O(\log^2 n)$ time for a 2-dimensional problem, which becomes $O(\log^d n)$ for a $d$-dimensional situation.
\end{abstract}


\newpage

\section{Introduction}
Range queries appear frequently in different problems in Computer Science. Many interesting and useful problems can also be reduced to range queries. For example, the problem of computing the Lowest Common Ancestor (LCA) of two nodes in a tree can be reduced to a range minimum query problem \cite{bender2000lca} to solve it efficiently \cite{harel1984fast}. Also range updates and range queries have a lot of diversified applications in domains like database theory \cite{pagel1993towards,mykletun2006aggregation}, sensor networks \cite{li2003multi}, image processing \cite{deng2001efficient}, cryptography \cite{boneh2007conjunctive}, computational geometry \cite{de1997computational,easwarakumar2015bits}, geographic information systems \cite{samet1984geographic} etc.

In its simplest form, the dynamic range query problem involves two operations, namely, query and update. Suppose we are considering range sum queries and updates; then a query for a given range will return the sum of all the elements within the supplied range. On the other hand, an update in this context adds a specific value (given with the update request) to each of the elements within the range. In this context, the aim is to construct a data structure that can be used to efficiently answer queries and handle updates.

The Segment Tree \cite{de1997computational} data structure is particularly notable in handling the range query and update operations. A Segment Tree is a complete binary tree where the nodes operate on segments. These segments are divided into two equal or near equal parts recursively and are merged later on to perform dynamic range query operations efficiently in logarithmic time. Compared to other similar data structures like Fenwick Trees \cite{fenwick1994new}, Segment Trees are more versatile, which has led to diversified applications of this tree-based data structure. Segment Trees have been used in different topics in Computer Science including but not limited to networking \cite{zheng2006distributed, gupta2001algorithms}, computer vision \cite{mei2013segment}, computational geometry \cite{bentley1980optimal,chazelle1994algorithms} etc. to name a few.

However, if we consider update operations, Segment Trees have only been able to enjoy success in a single dimension. In higher dimensions, it suffers greatly due to its inability to handle range updates efficiently therein. In single dimension, Segment Trees exploit the idea of Lazy Propagation that allows it to perform range updates in logarithmic time \cite{halim2013competitive}. But unfortunately, this trick does not generalize to higher dimensions. In spite of its heavy use in solving different problems in the literature, we do not find much works therein focusing on improving the Segment Tree itself. In this paper, we have made an attempt to improve the performance of the update operation of Segment Trees in higher dimensions. In particular, we propose and incorporate a novel idea of scaled update and partial query that in combination with the concept of lazy propagation equip Segment Trees with the capability to perform range update operations in two-dimensional context efficiently (Section \ref{sec:res}). We also discuss how our approach generalizes to higher ($> 2$) dimensions (Section \ref{sec:dis}). To the best of our knowledge, this is the first (successful) attempt to achieve this feat and we believe that our modified Segment Tree data structure will be extremely useful in diversified applications.

\section{Related Works}
\label{sec:rela}

Dynamic range queries have been a very popular problem in the algorithmic community. Lueker \cite{lueker1978data} proposed a data structure that performs orthogonal range query for some given points in a multidimensional space. However, this orthogonal range query was limited to returning the number of points within a specified range, updated through inserting and deleting points. Lueker later adopted a different transformation on decomposable online searching problems \cite{lueker1979transformation}. The data structure proposed by Willard \cite{willard1985new} included the range sum operation as well. Fredman \cite{fredman1981lower} explored the lower bound of such operations and proved that a sequence of $n$ insertions, deletions and range queries intermixed together demands a lower bound of $O(n(logn)^d)$, where $d$ denotes the dimensionality of the space.

However, the problem we are concerned with deals with a special version of this problem. Instead of working in a Cartesian space, our motivation lies in array operations. The complexity involved in offline partial sum computation was discussed by Chazelle \cite{chazelle1991complexity}. The dynamic range sum query problem has been sufficiently addressed by Fredman \cite{fredman1982complexity}, where he described ways to perform dynamic range sum query operations on a 1-dimensional array. However, his update operation was restricted to point-updates, i.e., updating a single value by adding a constant with it instead of updating a range which we are concerned with. He demonstrated how a balanced binary tree covering the various ranges, which is commonly termed as Segment Tree \cite{de1997computational} by the folklore, can be utilized to solve such problem in logarithmic time. He, however, presented an open question regarding the generalization of the model for higher dimensional spaces. Toni \cite{toni2003average} similarly performed range sum queries with point-update, albeit using a graph data structure. Tight bounds of partial sums were further discovered by \cite{puaatracscu2004tight}. Fenwick also proposed a novel data structure to perform range sum queries \cite{fenwick1994new}. However, none of these work accounted for updates in a range. 

Similar to the folklore solution of using balanced binary tree i.e., Segment Tree to solve point-update based dynamic range sum query problem, another technique became very popular to perform range sum queries even after updating ranges instead of a single index. This method involves Segment Trees in managing the range updates and computing the online queries using a novel idea called Lazy Propagation (please refer to Section \ref{sec:preli} and Section \ref{app:lazy}). However, this approach is limited to a single dimension only, which prevents us from solving the dynamic range sum query for more than one dimensions using Segment Trees. Mishra \cite{mishra2013new} elaborates on this fact by proving how higher dimensional Segment Trees cannot leverage the benefits of lazy propagation. In the same work, Mishra also proposes some modifications to higher dimensional Fenwick trees to compute dynamic range sum query on a multidimensional space.

\section{Preliminaries}
\label{sec:preli}

\subsection{Problem Specification}
In this paper we primarily discuss range sum queries and updates over a two dimensional array. A sub-array of a two dimensional array $A$ is $A[x_1:x_2 , y_1:y_2]$, where $x_1$,
$x_2$ are coordinates along the 1st dimension ($x$ dimension), and $y_1$, $y_2$ are coordinates along the 2nd dimension ($y$ dimension). It consists of all the elements $A[x,y]$, such that $x_1 \leq x \leq x_2$ and $y_1 \leq y \leq y_2$. The $query$ operation on the aforementioned sub-array returns the sum of all the elements within the sub-array:

$$ query(A[x_1:x_2,y_1:y_2]) = \sum_{x=x_1}^{x_2} \sum_{y=y_1}^{y_2} A[x,y] $$

The $update$ operation, on the other hand, adds a constant value $c$ to all the elements of the sub-array:

$$ update(A[x_1:x_2,y_1:y_2],c) \rightarrow A[x,y] = A[x,y] + c ,\forall x_1 \leq x \leq x_2, y_1 \leq y \leq y_2 $$

\subsection{The Segment Tree Data structure}
A Segment Tree can perform range queries and updates in logarithmic time following a divide and conquer approach \cite{de1997computational}. It is defined on a segment, which is recursively divided into two smaller segments and merged together to compute aggregate functions on the segments. Although Segment Tree following this approach can only perform point updates, the concept of `Lazy Propagation' allows the data structure to perform range queries and updates in logarithmic time. A Segment Tree can be extended to higher dimensions by successively cascading it. In a 2D Segment Tree, each of the nodes of the Segment Tree contains another Segment Tree therein. Thus, two segments along the two dimensions are simultaneously considered. However, lazy propagation concept does not generalize to higher dimensions and hence 2D (or higher dimensional) Segment Trees cannot do better than $O(n^{d-1}\log n)$ time query and update operations, where $n$ is the maximum size of the array in any one dimension and $d$ is the number of dimensions \cite{mishra2013new}. The operations of Segment Tree are illustrated in Figure \ref{fig:1dseg}. A more comprehensive overview of Segment Trees can be found in Appendix \ref{sec:aover}.

\begin{figure}[h]
    \centering
    \begin{subfigure}[b]{0.3\textwidth}
        \includegraphics[width=\textwidth]{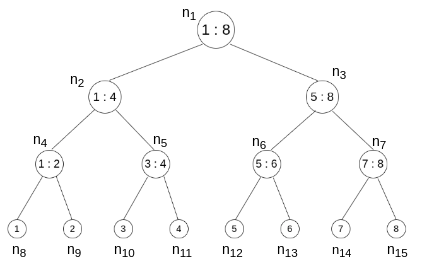}
        \caption{Segment Tree}
        \label{fig:1dseg1}
    \end{subfigure}
    ~ 
    \begin{subfigure}[b]{0.3\textwidth}
        \includegraphics[width=\textwidth]{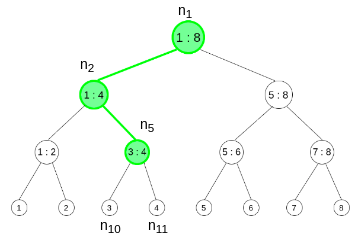}
        \caption{Update(3,4)}
        \label{fig:1dseg2}
    \end{subfigure}
    ~ 
    \begin{subfigure}[b]{0.3\textwidth}
        \includegraphics[width=\textwidth]{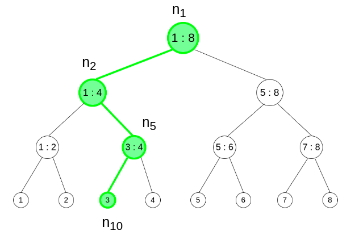}
        \caption{Query(3,3)}
        \label{fig:1dseg3}
    \end{subfigure}
    \caption{1 Dimensional Segment Tree. \ref{fig:1dseg1} shows the structure of the Segment Tree, which is a complete binary tree and all the nodes are defined on a segment, while the leaf nodes work on individual indices. \ref{fig:1dseg2} illustrates the lazy propagation operation. While updating the range (3,4), upon entering the node operating on (3,4) range, we update lazy values and backtrack instead of traversing to the individual leaf nodes. \ref{fig:1dseg3} demonstrates the query operation. While querying on the index 3, we keep dividing the regions into 2 segments until the node is reached. Along the way, the lazy values are passed and finally, the value from the node is retrieved, combining them we obtain the result of the query. }\label{fig:1dseg}
\end{figure}

\subsection{Terminologies}
\label{sec:termino}
For the rest of the paper, we denote the 1st (2nd) dimension as $x$-dimension ($y$-dimension) or simply $x~(y)$. We use the term \textit{region} to mean the sub-array of the array spanned by that region, i.e., indices of the elements. Also by $\log n$ we imply $\log_2n$. We also define the terms $x$-superregion and $x$-subregion as follows.
\begin{defn}
A region $r \equiv [x_1:x_2,y_1:y_2]$ is a $x$-superregion of a region $r'\equiv [x'_1:x'_2,y'_1:y'_2]$, if the ranges along $y$-dimension are equal in both cases (i.e., $y_1=y'_1, y_2=y'_2$ ), but the range along $x$-dimension of $r'$ is a proper subrange of that of $r$ ($[x_1:x_2] \cup [x'_1:x'_2] = [x_1:x_2], [x_1:x_2] \cap [x'_1:x'_2] = [x'_1:x'_2], [x_1:x_2] \neq [x'_1:x'_2]$). On the other hand, in the above, $r'$ is said to be a $x$-subregion of the region $r$.
\end{defn}
In our approach, we propose and implement two different types of updates in the Segment Tree data structure (See Section \ref{sec:algoupdate} for details) as defined below.
\begin{defn}[Dispersed and Intended Updates.] Consider an update query on a region $r\equiv [x_{start}:x_{end},y_{start}:y_{end}]$. Now suppose, $r$ is a $x$-subregion of $r' \equiv [x_1':x_2',y_1':y_2']$. Also suppose that $r$ is either a $x$-superregion of $r'' \equiv [x_1'':x_2'',y_1'':y_2'']$, or completely matches with $r''$. If a node is operating on the region $r'~(r'')$, then this is a `Dispersed Update' (`Intended Update').
\end{defn}

Furthermore, we classify the queries into two classes as defined below.
\begin{defn}[Partial and Complete Queries.]
While querying on a range $r \equiv [x_{start}:x_{end},y_{start}:y_{end}]$, we perform a Partial (Complete) Query on a node operating on range $r'\equiv [x_1:x_2,y_1:y_2]$ such that $r$ is a $x$-subregion ($x$-superregion) of $r'$ (or completely matches with $r'$).
\end{defn}

In our proposed 2D Segment Tree, at each node, we store two types of values and lazy updates, namely, Global and Local. Global values and global lazy updates propagate along the first dimension. These are the results of Intended Updates and are used both in Partial and Complete Queries. Local values and local lazy updates, on the contrary, remain confined within the first dimension and they do not propagate. They account for Dispersed Updates and are only considered during Complete Query. In case of Partial Queries, we dilute, i.e., scale down the values stored in the nodes by scaling them, and we only consider the global values stored in the nodes. In the cases of Complete Queries, we do not dilute the values stored in the nodes and consider both local and global values. We discuss the details in the next section.

In what follows we will be using the following notations with respect to a Segment Tree. A node $node$ of a Segment Tree is defined by the following attributes:
\begin{itemize}
    \item $node.range$: The range of the node specifies the segment it is operating on.
    \item $node.size$: The size denotes the length, i.e., the number of elements in the segment it is operating on.
    \item $node.value$: The value of a node equals the result of the function of consideration when operated on the segment the node is defined on. For example, for our range sum query problem, the value represents the sum of all the elements in the segment.
    \item $node.lazy$: The value of the lazy update of a node signifies that all the descendants of that node should be updated by this value implicitly.
\end{itemize}
In our proposed modification we decompose the values and lazy updates into two components `Global' and `Local'. Therefore these quantities are denoted by $node.global.value$, $node.global.lazy$, $node.local.value$, $node.local.lazy$ for Global and Local components respectively.

\section{Main Results}
\label{sec:res}
In this section we present our proposed approach, prove its correctness and analyze the time and space complexity thereof.

\subsection{Proposed Algorithm}
\label{sec:proalg}
In what follows we assume that we are handling the dynamic range sum query problem over a two-dimensional array $A$ of size $n \times m$.
\label{sec:algo}

\subsubsection{Construction}
The construction of our proposed 2D Segment Tree is almost identical to that of the original 2D Segment Tree. However, in our proposed Segment Tree, at every node, we store two types of values, namely, local values and global values (as discussed in the previous section). The reason for this and further details thereof will be spelled out shortly.

\subsubsection{Update Operations}
\label{sec:algoupdate}
The goal of an \emph{update} operation in our context is to update a region $r \equiv [x_{start}:x_{end},y_{start}:y_{end}]$ by adding a constant value $c$ (supplied with the update request) to each of the elements of the array $A$ residing within the region $r$. Similar to the original 2D Segment Tree algorithm, our modified algorithm first divides the regions along the first dimension ($x$ dimension), and then starts breaking the regions along the second dimension ($y$ dimension). The update algorithm starts from the root node $node_R$, which is defined on the entire range along the $x$ dimension, i.e., $[1:n]$. If the region $r \equiv [x_{start}:x_{end},y_{start}:y_{end}]$ covers the entire range along the $x$ dimension, i.e., $x_{start}=1,~x_{end}=n$, we do not break the regions along the $x$ dimension; rather we go inside the root node $node_R$, which itself contains a 1D Segment Tree $T_{node_R}$. We then start dividing the range along the $y$ dimension following the classical 1D Segment Tree algorithm. We again start from the root node of the 1D Segment Tree, $T_{node_R}$, that accounts for all the regions whose range along $x$ dimension is $[1:n]$. We then start breaking the regions along $y$ dimension until we obtain a node $node_{R_j}$, that covers the region $r$ or its subregions. Then we update the global values of the node $node_{R_j}$ and set the global lazy value as shown in Equations \eqref{eqn:globupd} and \eqref{eqn:globupdlazy} below.

\begin{equation}
\label{eqn:globupd}
    node_{R_j}.global.value = node_{R_j}.global.value + node_{R_j}.size \times c
\end{equation}
\begin{equation}\label{eqn:globupdlazy}
    node_{R_j}.global.lazy = node_{R_j}.global.lazy + c
\end{equation}

After updating the node $node_{R_j}$, as we backtrack to the root of $T_{node_R}$, along the way for each node $node_{R_{j'}}$ we modify its global value as follows (Equation \eqref{eqn:globupd}): we assign the sum of (a) the global values of its left child $node_{R_{j'|l}}$ and right child $node_{R_{j'|r}}$ and (b) the product of global lazy values and the size of the node.

\begin{equation}\label{eqn:globupdbacktrack}
\begin{split}
    node_{R_{j'}}.global.value = node_{R_{j'|l}}.global.value + node_{R_{j'|r}}.global.value \\ + node_{R_{j'}}.global.lazy \times node_{R_{j'}}.size
\end{split}
\end{equation}

Here we are updating the global values stored in each node, and as we mentioned earlier in Section \ref{sec:termino} we term them as \emph{Intended Updates}. Now, it may be the case that the intended update region $r$ along the $x$ dimension does not span the entire range $[1:n]$. In such a scenario we divide the region along the $x$ dimension further until we arrive at a node $node_i$ that operates on a region $[x_1:x_2]$ such that it is completely contained inside the range of $r$ along the $x$ dimension, $[x_{start}:x_{end}]$, i.e., $x_1 \geq x_{start}$ and $x_2 \leq x_{end}$. Similar to the root node, the node $node_i$ encapsulates another 1D Segment Tree $T_{node_i}$ that operates on all the regions that span the range $[x_1:x_2]$ completely, i.e., the root node of tree $T_{node_i}$ is defined on the region $[x_1:x_2,1:m]$. Then we perform the same operations discussed above for the tree $T_{node_R}$.

As we are dividing the regions along $x$ dimension, not all divided regions are subregions of the intended update region. On the other hand, the intended update region can be a subregion of some divided regions. There may also exist some regions that overlap with the update region and some regions that are completely disjoint to it. Similar to the classical Segment Tree update algorithm, we can safely ignore the disjoint regions. For updating a region $r' \equiv [x'_1:x'_2]$ having overlaps with $r$, we first trim down the intended update region $r$ along $x$ dimension to get the trimmed range $[x'_{start}:x'_{end}]$, such that $x'_{start} = \max(x_{start},x'_1)$ and $x'_{end} = \min(x_{end},x'_2)$. Clearly, if the range of $x$ dimension of $r$ is contained completely within that of $r'$, the trimmed region stays the same. Thus, for the ease of implementation, we perform this trimming in both the cases.

For these cases, after dividing along $x$ dimension, we start dividing the regions along $y$ dimension following the same rules applied earlier. However, we now perform a \emph{dispersed update} as follows. Let, we intend to update the region $r$, and we are currently at a node $node$ operating on $r'$. Now let us assume that $r$ is either a subregion of $r'$ or it intersects with $r'$. As stated in the previous paragraph we perform a trimming on $r$ as a general rule. Then, we distribute the effect of updating $r$ over the whole region $r'$ and use a scaled value $c'$ (instead of $c$) as defined below:

\begin{equation}
        c' = scaling \times c = \frac{x'_{end} - x'_{start} + 1 }  {x_2-x_1+1} \times c
\end{equation}

Here, $[x'_{start}:x'_{end}]$ is the trimmed range (as discussed above); $[x_1:x_2]$ is the range covered by the region along $x$ dimension; recall that $c$ is the value supplied with the (original) update request.

This scaled value $c'$ is then dispersed in the entire region through our \emph{dispersed update} that updates the local values of the nodes as follows.

\begin{equation}
    node_{i_j}.local.value = node_{i_j}.local.value + node_{i_j}.size \times c'
\end{equation}
\begin{equation}
    node_{i_j}.local.lazy = node_{i_j}.local.lazy + c'
\end{equation}

Similar to the global values, the local values of the parent nodes are also updated as we backtrack to the root node. Suppose that during a stage of backtracking we are at node $node_{i_j}$ having $node_{i_j|l}$ and $node_{i_j|r}$ as its left and right child respectively. Then, the local value of node $node_{i_j}$ is updated as follows.

\begin{equation}
\begin{split}
    node_{i_j}.local.value = node_{i_j|l}.local.value + node_{i_j|r}.local.value \\ + node_{i_j}.local.lazy \times node_{i_j}.size
\end{split}    
\end{equation}

The pseudocode of the update operation is presented as Algorithm \ref{alg:update}.

\begin{algorithm}
\caption{Update}
\label{alg:update}
    \begin{algorithmic}
            
        \State \Call{UpdateByX}{$rootNode$, $1$, $n$, $x_{start}$, $x_{end}$, $y_{start}$, $y_{end}$, $c$}
         \Comment{Starting from the root node}
        \State
        
        \Function{UpdateByX}{$node$, $x_1$, $x_2$, $x_{start}$, $x_{end}$, $y_{start}$, $y_{end}$, $c$}
        
            \If{$(x_{1}\,:\,x_{2})$ is within $(x_{Start}\,:\,x_{End})$}  \Comment{Intended Update}
                \State \Call{UpdateByY}{$newNode$, $x_1$, $x_2$, $1$, $m$, $x_{start}$, $x_{end}$, $y_{start}$, $y_{end}$, $c$}
        
            \ElsIf{$(x_{1}\,:\,x_{2})$ is outside $(x_{Start}\,:\,x_{End})$}     \Comment{Disjoint Region}
                \State \Return 
            \Else
                \State $x_{mid} = \frac{x_1+x_2}{2}$
                \State \Call{UpdateByX}{$leftChild$, $x_1$, $x_{mid}$, $x_{start}$, $x_{end}$, $y_{start}$, $y_{end}$, $c$}
                
                \State \Call{UpdateByX}{$rightChild$, $x_{mid}+1$, $x_{2}$, $x_{start}$, $x_{end}$, $y_{start}$, $y_{end}$, $c$}
                
                \State $x'_{start} = max(x_{start},\;x_1)$
                 \Comment{Trimming}
                \State $x'_{end} = min(x_{end},\;x_2)$
                
                \State $scaling=\frac{x'_{end}-x'_{start}+1}{x_2-x_1+1}$
                 \Comment{Dispersed Update}
                \State \Call{UpdateByY}{$newNode$, $x_1$, $x_2$, $1$, $m$, $x'_{start}$, $x'_{end}$, $y_{start}$, $y_{end}$, $c \times scaling$}
            
            \EndIf
        \EndFunction
        \State
        
        \Function{UpdateByY}{$node$, $x_1$, $x_2$, $y_1$, $y_2$, $x_{start}$, $x_{end}$, $y_{start}$, $y_{end}$, $c$}
    
            \If{$(y_{1}\,:\,y_{2})$ is within $(y_{Start}\,:\,y_{End})$}
    
                \If{$(x_{1}\,:\,x_{2})$ is within $(x_{Start}\,:\,x_{End})$}
                 \Comment{Intended Update}
                    \State $node.global.value = node.global.value + c \times (x_{2}-x_{1}+1) \times (y_{2}-y_{1}+1)$
                    \State $node.global.lazy = node.global.lazy + c$
                \Else  \Comment{Dispersed Update}
                    \State $node.local.value = node.local.value + c \times (x_{2}-x_{1}+1) \times (y_{2}-y_{1}+1)$
                    \State $node.local.lazy = node.local.lazy + c$
                \EndIf
            
            \ElsIf{$(y_{1}\,:\,y_{2})$ is outside $(y_{Start}\,:\,y_{End})$}
             \Comment{Disjoint Region}
                \State \Return 
    
            \Else
                \State $y_{mid} = \frac{y_1+y_2}{2}$
                
                \State \Call{UpdateByY}{$leftChild$, $x_1$, $x_2$, $y_1$, $y_{mid}$, $x_{start}$, $x_{end}$, $y_{start}$, $y_{end}$, $c$}
                
                \State \Call{UpdateByY}{$rightChild$, $x_1$, $x_2$, $y_{mid}+1$, $y_2$, $x_{start}$, $x_{end}$, $y_{start}$, $y_{end}$, $c$}

                \State $node.local.value = leftChild.local.value + rightChild.local.value + node.local.lazy \times (x_{2}-x_{1}+1) \times (y_{2}-y_{1}+1) $
                
                \State $node.global.value = leftChild.global.value + rightChild.global.value + node.global.lazy \times (x_{2}-x_{1}+1) \times (y_{2}-y_{1}+1) $
            \EndIf
        \EndFunction

    \end{algorithmic}
\end{algorithm}


\subsubsection{Query Operation}
\paragraph{4.1.3.1 Partial Query}

Recall that the query operation must return the sum of all the elements in the supplied range/region $r \equiv [x_{start}:x_{end},y_{start}:y_{end}]$. To do that, our algorithm starts from the root node $node_R$. We start dividing the regions along $x$ dimension until the query range along $x$ dimension (i.e., $[x_{start}:x_{end}]$) is obtained. Note that, as we travel from one node to another, while dividing the regions along $x$ dimension, even in cases where ranges along $x$ dimension mismatch, we perform a query on the Segment Tree inside each node (say node $node_i$). This query is termed as `Partial Query', where we only query the global value in such cases as follows. Moreover, we dilute, i.e., scale down the returned values (i.e., the global values stored in the nodes) so as to consider the actual contribution of the region represented by the node $node_i$, in such cases as follows. If we are at a node $node_i$ that operates on the region $[x_1:x_2,y_1:y_2]$, then, the scaling factor will be as follows.

\begin{equation}
    scaling_{node_i} = \frac{x_{end} - x_{start} + 1}{x_2 - x_1 +1}
\end{equation}

Similar to the proposed update operation, the query region is trimmed such that only the query region that falls within the region operated by the node $node_i$ is considered. In addition to considering the global values stored in node $node_i$, we also propagate the global lazy values stored therein. These lazy values do not require scaling as they are multiplied by the size of the query region. The partial result obtained from node $node_i$ is therefore as follows.
%

\begin{equation}
\begin{split}
   partial\_query_{node_i} = node_i.global.value \times scaling_{node_i} \\+ \sum_{node \in ancestors(node_i)}^{ } node.global.lazy \times query\_region.size
\end{split}
\end{equation}

\paragraph{4.1.3.2 Complete Query}
On the other hand, while dividing the regions along $x$ dimension, if we obtain a region $r'$ such that the range along $x$ dimension is either equal or is contained within that of the query region, we perform a `Complete Query'. For `Complete Query' we consider both the local and the global values. Similarly, both the local and global lazy updates are propagated. Moreover, no scaling is performed as the region already corresponds to the query region. The complete result obtained from a node $node_i$ is therefore as follows.


\begin{equation}
\begin{split}
Complete\_Query_{node_i} = node_i.global.value + node_i.local.value + \\ 
\sum_{node \in ancestors(node_i)}^{ } (node.global.lazy + node.local.lazy) \times query\_region.size
 \end{split}
\end{equation}

\paragraph{4.1.3.3 Combining}

After performing both the `Partial Query' and `Complete Query' the results are back-propagated to the root node. From there we compute the output of query operation by combining, i.e., adding all these results together.

\begin{equation}
    Query([x_{start}:x_{end},y_{start}:y_{end}]) = \sum_{node_p \in N} partial\_query(node_p) + \sum_{node_C \in N{_x}} Complete\_Query(node_C)
\end{equation}

Here, $N_x$ contains all visited nodes operating on $r$ or regions that are $x$-subregion of $r$. The set $N$ comprises the rest of the visited nodes. The pseudocode of the query operation is presented as Algorithm \ref{alg:query}.

\begin{algorithm}
\caption{Query}
\label{alg:query}
    \begin{algorithmic}
            
        \State \Call{QueryByX}{$node$, $1$, $n$, $x_{start}$, $x_{end}$, $y_{start}$, $y_{end}$}
        \Comment{Starting from the root node}
        \State
        
        \Function{QueryByX}{$node$, $x_1$, $x_2$, $x_{start}$, $x_{end}$, $y_{start}$, $y_{end}$}
        
            \If{$(x_{1}\,:\,x_{2})$ is within $(x_{Start}\,:\,x_{End})$} \Comment{Complete Query}
                \State {\Return \Call{QueryByY}{$newNode$, $x_1$, $x_2$, $1$, $m$, $x_{start}$, $x_{end}$, $y_{start}$, $y_{end}$, $0$}}
        
            \ElsIf{$(x_{1}\,:\,x_{2})$ is outside $(x_{Start}\,:\,x_{End})$}   \Comment{Disjoint Region} 
                \State \Return 0
            \Else
                \State $x_{mid} = \frac{x_1+x_2}{2}$
                \State $x'_{start} = max(x_{start},\;x_1)$
                \Comment{Trimming}
                \State $x'_{end} = min(x_{end},\;x_2)$
                \State \Return \Call{QueryByY}{$newNode$, $x_1$, $x_2$, $1$, $m$, $x'_{start}$, $x'_{end}$, $y_{start}$, $y_{end}$, $0$}
                \State \Comment{Partial Query}
                \State \quad\qquad + \Call{QueryByX}{$leftChild$, $x_1$, $x_{mid}$, $x_{start}$, $x_{end}$, $y_{start}$, $y_{end}$}
                \State \quad\qquad + \Call{QueryByX}{$rightChild$, $x_{mid}+1$, $x_{2}$, $x_{start}$, $x_{end}$, $y_{start}$, $y_{end}$}
            \EndIf
        \EndFunction
        \State
        
        \Function{QueryByY}{$node$, $x_1$, $x_2$, $y_1$, $y_2$, $x_{start}$, $x_{end}$, $y_{start}$, $y_{end}$, $lazy$}
    
            \If{$(y_{1}\,:\,y_{2})$ is within $(y_{Start}\,:\,y_{End})$}
    
                \If{$(x_{1}\,:\,x_{2})$ is within $(x_{Start}\,:\,x_{End})$} \Comment{Complete Query}
                    \State
                    \Return $node.local.value + node.global.value + lazy \times (x_{2}-x_{1}+1) \times (y_{2}-y_{1}+1)$
                \Else \Comment{Partial Query}
                    \State $scaling = \frac{(x_{end}-x_{start}+1)}{(x_{2}-x_{1}+1)}$
                    \Comment{Diluting values}
                    \State \Return $node.global.value \times scaling + lazy \times (x_{end}-x_{start}+1) \times (y_{2}-y_{1}+1)$ 
                \EndIf
        
            \ElsIf{$(y_{1}\,:\,y_{2})$ is outside $(y_{Start}\,:\,y_{End})$}
            \Comment{Disjoint Region}
                \State \Return 0
    
            \Else
                \State $y_{mid} = \frac{y_1+y_2}{2}$
    
                \If{$(x_{1}\,:\,x_{2})$ is within $(x_{Start}\,:\,x_{End})$}
                \Comment{Complete Query}
                    \State \Return \Call{QueryByY}{$leftChild$, $x_1$, $x_2$, $y_1$, $y_{mid}$, $x_{start}$, $x_{end}$, $y_{start}$, $y_{end}$, $lazy+node.global.lazy+node.local.lazy$ }
                    \State \quad \qquad $+$\Call{QueryByY}{$rightChild$, $x_1$, $x_2$, $y_{mid}+1$, $y_2$, $x_{start}$, $x_{end}$, $y_{start}$, $y_{end}$, $lazy+node.global.lazy+node.local.lazy$ }
            
                \Else \Comment{Partial Query}
                    \State \Return \Call{QueryByY}{$leftChild$, $x_1$, $x_2$, $y_1$, $y_{mid}$, $x_{start}$, $x_{end}$, $y_{start}$, $y_{end}$, $lazy+node.global.lazy$ }
                    \State \quad \qquad $+$\Call{QueryByY}{$rightChild$, $x_1$, $x_2$, $y_{mid}+1$, $y_2$, $x_{start}$, $x_{end}$, $y_{start}$, $y_{end}$, $lazy+node.global.lazy$ }
                \EndIf
            \EndIf
        \EndFunction

    \end{algorithmic}
\end{algorithm}

\subsection{Correctness}
\label{sec:corr}
\begin{lemma}
\label{lemma:updglob}
`Intended updates' update the regions as well as their subregions precisely.
\end{lemma}
\begin{proof}
From the steps of our proposed update algorithm (Section \ref{sec:algoupdate}), it is trivial to show that the algorithm updates the actual regions correctly. In all the situations, the goal of the proposed algorithm is to keep on dividing the ranges until the actual update region or its subregion is obtained. After reaching that, the algorithm updates the global values stored in the node and starts backtracking to the root node. In cases when the algorithm reaches the actual region but not its subregion, the subregions are still updated as the updates are stored in the global values which are passed to the decedent subregions anyway. Hence, it is ensured that the `Intended Updates' not only precisely update the actual region, but also its subregions.
\end{proof}

From the definition of 1D Segment Tree, while updating a region $r$ with a value $c$, a node $node$ operating on a subregion $r'$ is updated using the following rule:
\begin{equation}
\label{eqn:lemma3.2.1}
    {node.value} = {node.value} +{c} \times {node.size}
\end{equation}
Recall that $node.size$ refers to the size of the region the node is operating on, i.e., $node.size = ||r'||$. Unfortunately, when working on 2 or higher dimensions, it has not been possible to perform such operations along all but the innermost dimension. This motivated us to propose relaxed updates along $x$ dimension, which we termed as `Dispersed Updates'. Now we have the following lemma.

\begin{lemma}
\label{lemma:dspUpd}
The effect of updating a proper subregion of a region is realized by `Dispersed Updates'.
\end{lemma}
\begin{proof}
Recall that, unlike the $x$ dimension the ranges along $y$ dimension are broken completely while performing the `Dispersed Updates'. Without any loss of generality we can consider a proper subregion as a proper $x$-subregion. Suppose we intend to update $r \equiv [x_{start}:x_{end},y_{start}:y_{end}]$, and the node $node$ is associated to $r' \equiv [x_1:x_2,y_{start}:y_{end}]$, which is a proper $x$-superregion of $r$. So, we have $r' \cup r = r', r' \cap r = r, r \neq r'$. When updating the node $node$ by adding $c$, Equation \eqref{eqn:lemma3.2.1} suggests to add $c$ to all the $||r'||$ elements of $r'$. But in reality, only $||r||$ number of elements are actually being updated. Thus the node value erroneously increases by $c \times ||r'||$, differing from the expected increment of $c \times ||r||$.
Now, `Dispersed Update' scales the update value $c$ to $c' = c \times s$ by a ratio $s = \frac{||r||}{||r'||}$ ,  thereby updating the values correctly.
%
\end{proof}

Recall that, a dispersed update only modifies the local values of a node $node$ (i.e.,  $node.local.value$ and $node.local.lazy$). Lemma \ref{lemma:dspUpd} claims that a dispersed update accurately captures the modifications of the subregions beneath it. Hence we have the following straightforward corollary.

\begin{corollary}
Local values (both $value$ and $lazy$) link a node to the updates of its descendants.
\end{corollary}

Recall that while modifying a node $node$ operating on a region $r' \equiv [x_1:x_2,y_1:y_2]$ that intersects with the region of update $r \equiv [x_{start}:x_{end},y_{start}:y_{end}]$, our proposed algorithm trims $r$ to $[x'_{start}:x'_{end},y_{start}:y_{end}]$ such that $ x'_{start} = max(x_{start},x_1)$ and $x'_{end} = min(x_{end},x_2)$. Following the arguments from lemma \ref{lemma:dspUpd} we have the following corollary.
\begin{corollary}
\label{cor:intersect}
Trimming down the update regions (as mentioned above) allows the intersecting regions to be updated properly.
\end{corollary}

\begin{lemma}
\label{lem:dilute}
Diluting the values ensures the contribution of a region to its proper subregion.
\end{lemma}
\begin{proof}
Suppose we are considering a node $node$, defined on region $r'\equiv [x_1:x_2,y_1:y_2]$, whereas the region of query $r\equiv [x_{start}:x_{end},y_{start}:y_{end}]$ is a proper subregion of $r'$. Suppose, $node.global.value = v$. This accounts for updating all the elements contained in that region. However, this value, $v$, is the representative of all $||r'||$ elements encompassed by the region $r'$. But, we are interested in only $||r||$ elements of them belonging to the region $r$. Hence, returning the global value $v$ wrongfully over-estimates the actual value. Therefore, we dilute the actual value stored in a node $node$ by a factor as follows:

\begin{equation}
    scaling = \frac{x_{end}-x_{start}+1}{x_2-x_1+1} = \frac{||r||}{||r'||}
\end{equation}
\begin{equation}
    \textit{diluted value} = scaling \times node.global.value  = \frac{||r||}{||r'||} \times node.global.value
\end{equation}
Thus we are able to consider the updates of a region while querying on its subregion.
\end{proof}
We remark that while considering the lazy updates since we multiply them by the number of elements in the query region, $||r||$, it estimates the actual value. Now, in the lazy propagation process of the classical 1D Segment Tree, updates that are to be performed on a number of adjacent segments are stored in an ancestor node of the tree that represents a bigger segment covering the segments to be updated. During query time these updates are passed down and distributed properly such that the effects of the updates are fulfilled \cite{halim2013competitive}. Unfortunately, following the classical Segment Tree algorithm, it is only possible to implement lazy propagation in the innermost dimension, i.e., $y$ dimension \cite{mishra2013new}. As a result, it is impossible to retrieve the updates that were imposed on regions which may overlap with the query region but extends further along the $x$ dimension. Now we present and prove the following lemma.
\begin{lemma}
\label{lem:partial}
`Partial Query' mimics the lazy propagation procedure along $x$ dimension.
\end{lemma}
\begin{proof}
Partial query only considers the global values ($value$ and $lazy$) of a node and repeatedly perform queries until the actual query region is reached. Consider a node $node$ associated to a region $[x_1:x_2 , y_1:y_2]$. Suppose, a number of ancestors of this node are updated. In order to ensure efficiency, the update algorithm terminates after modifying those parent nodes but not node $node$. Thus, while querying on node $node$, it is imperative to retrieve this information.

Since these updates were `Intended Updates' the global values were altered. By Lemma \ref{lem:dilute}, through diluting the values the actual updates of the region caused by the encompassing regions can be obtained. Therefore, by repeatedly performing this `Partial Query' for every node visited while dividing the region along $x$ dimension, all the updates that were performed on the $x$-superregions are obtained. Thus, it emulates the lazy propagation scheme along the $x$ dimension.
\end{proof}

While performing the `Complete Query', we consider both the local and global values of a node $node$. The global values are results of `Intended Updates', which correspond to the updates made on that node, i.e., region. On the other hand, local values originate from `Dispersed Updates' which capture the effects of updating the subregions of that region (Lemma \ref{lemma:dspUpd}). Hence, from the definitions and the proposed algorithm, it is evident that `Complete Query' acknowledges both the updates to a region and its subregions. So, we have the following lemma.

\begin{lemma}
\label{lem:completeQuery}
`Complete Query' not only considers the updates to a region but also the updates to the subregions of that region.\qed
\end{lemma}

\begin{theorem}
\label{thm:update}
The Update operation is correct.
\end{theorem}
\begin{proof}
In order to establish the correctness of the update operation, it suffices to ensure that when updating a specific region, the proposed update algorithm modifies all the nodes associated to that region appropriately. Suppose, we  are updating a region $[x_{start}:x_{end},\,y_{start}:y_{end}]$. Without any loss of generality, five types of regions can be specified as follows (also see Figure \ref{fig:updates}):
\begin{enumerate}
    \item $R_1$ is the actual update region.
    \item $R_2$ belongs to the superregions of $R_1$, i.e., regions that encompass $R_1$ as a subregion.
    \item $R_3$ belongs to the subregions of $R_1$.
    \item $R_4$ belongs to the regions that intersect with $R_1$.
    \item $R_5$ belongs to the regions that are disjoint with $R_1$.
\end{enumerate}
\begin{figure}[h]
  \centering
    \includegraphics[width=0.75\textwidth]{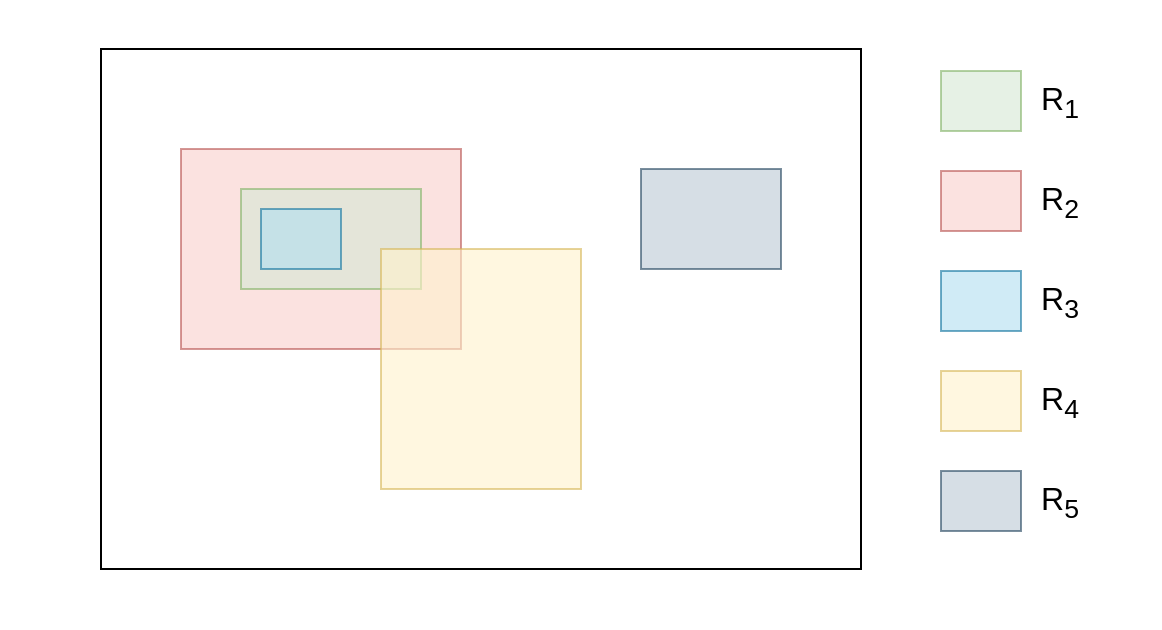}
    \caption{Different types of regions to consider while updating the region $R_1$. These include the superregions ($R_2$), subregions ($R_3$), intersecting regions ($R_4$) and completely disjoint regions ($R_5$)}
    \label{fig:updates}
\end{figure}
By definition, intended updates of the proposed algorithm update both the update region and its subregions properly ($R_1$ and $R_3$). Lemma \ref{lemma:dspUpd} demonstrated how the proposed `Dispersed Update' captures the outcome of updating the region $R_1$ in the superregions $R_2$. Corollary \ref{cor:intersect} extends this idea to the intersecting regions $R_4$. Finally we can safely omit the disjoint regions (i.e., $R_5$) completely (Section \ref{sec:algoupdate}). Hence the result follows.
\end{proof}

\begin{theorem}
\label{thm:query}
The Query operation is correct.
\end{theorem}
\begin{proof}
While querying on a region $r \equiv [x_{start}:x_{end}, y_{start}:y_{end}]$ , it is necessary that all the updates that were performed on regions associated with $r$ are compiled. Without any loss of generality, values and updates of 3 types of regions should be considered as follows (also see Figure \ref{fig:queries}):
\begin{enumerate}
    \item $R_1$ is the query region, i.e., the region of interest.
    \item $R_2 \in R_{sup}$, the set of the superregions of $R_1$.
    \item $R_3 \in R_{sub}$, the set of the subregions of $R_1$.
\end{enumerate}
\begin{figure}[h]
  \centering
    \includegraphics[width=0.5\textwidth]{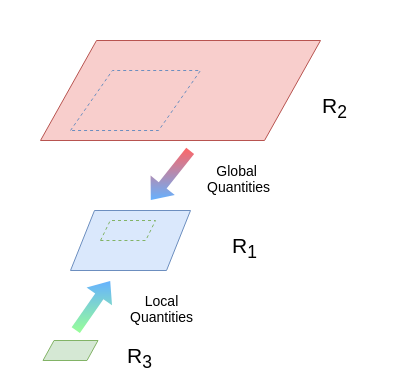}
    \caption{Different types of regions to consider while querying on a region $R_1$. Along with analyzing the region $R_1$ itself, it suffices to consider the superregions, $R_2$ (ancestors) and subregions, $R_3$ (descendents). The ancestors are connected though global values, whereas the descendents are linked by local values.}
    \label{fig:queries}
\end{figure}
Lemma \ref{lem:dilute} and Lemma \ref{lem:partial} show that the query algorithm compiles all the updates of the superregions of the region queried upon through `Partial Queries'. On the other hand, Lemma \ref{lem:completeQuery} proves the effectiveness of `Complete Query' in determining the values and the updates of the queried region and its subregions. Hence, as shown in the Figure \ref{fig:queries}, while querying on any region $[x_{start}:x_{end},y_{start}:y_{end}]$ the effects of updating the superregions are passed through global values ($value$ and $lazy$) and the impact of updating the subregions are compiled by combining the local values ($value$ and $lazy$). Therefore, the proposed query algorithm can query on any region correctly.
\end{proof}
Finally, based on the above arguments, particularly Theorem \ref{thm:update} and Theorem \ref{thm:query}, we can state the following:
\begin{theorem}[Correctness]
\label{thm:corr}
The proposed data structure handles range sum queries and dynamic updates correctly.\qed
\end{theorem}

%
%

\subsection{Analysis}
\label{sec:ana}
\subsubsection{Space Complexity}
A Segment Tree on an array of length $n$ is a binary tree with height $\ceil{\log n}$. 
The total number of nodes in a Segment Tree is $1 + 2 + 2^2 + 2^3 + ... + 2^{\ceil{\log n}}= 2\times 2^{\ceil{\log n}-1} \approx 2 \times n$. 
Now our proposed algorithm holds another Segment Tree inside each node. So, for a 2D array of size $n \times m$ the total space required is $O(2\times n \times 2\times m) = O(4\times n\times m) = O(n\times m)$.

\subsubsection{Time Complexity}
\label{sec:time}
1D Segment Tree (on a $n$-sized array) efficiently performs update and query operations in logarithmic time, i.e., $O(\log n)$ \cite{de1997computational}.
This follows trivially since the Segment Tree continuously divides the segments into two equal parts and a segment of length $n$ can be divided into two equal parts for at most $\log n$ time. 
Our proposed 2D Segment Tree similarly visits $O(\log n)$ number of nodes. But each of these nodes contain another 1D Segment Tree inside, incurring $O(\log m)$ time each. And this accounts for both repeated `Dispersed Updates' and `Partial Queries' at the nodes. Therefore, the overall time complexity becomes = $O(\log n \times \log m) = O(\log^2 n)$, assuming $n\geq m$, a huge improvement over the previous results of $O(n\log n)$ \cite{mishra2013new}.

\subsubsection{Experimental Study}
In order to validate the time complexity experimentally, we performed a number of random tests. We implemented 2D Segment Trees using our approach on two dimensional arrays of size $n \times n$ for $n=5$ to $900$. We then performed 100 random updates. Each update was followed by 100 random queries. The average time needed for the operations were calculated and plotted in a graph. We fit a $c \log^2 n$ curve and it was seen that for both update and query functions the time required follows this curve. These results are presented in Figure \ref{fig:time}. However, the times needed for updates seem a bit noisy, which is due to the fact that we took less number of samples for update operations (100) compared to that of the query operations ($100 \times 100 = 10000$).

\begin{figure}
    \centering
    \includegraphics[width=\textwidth]{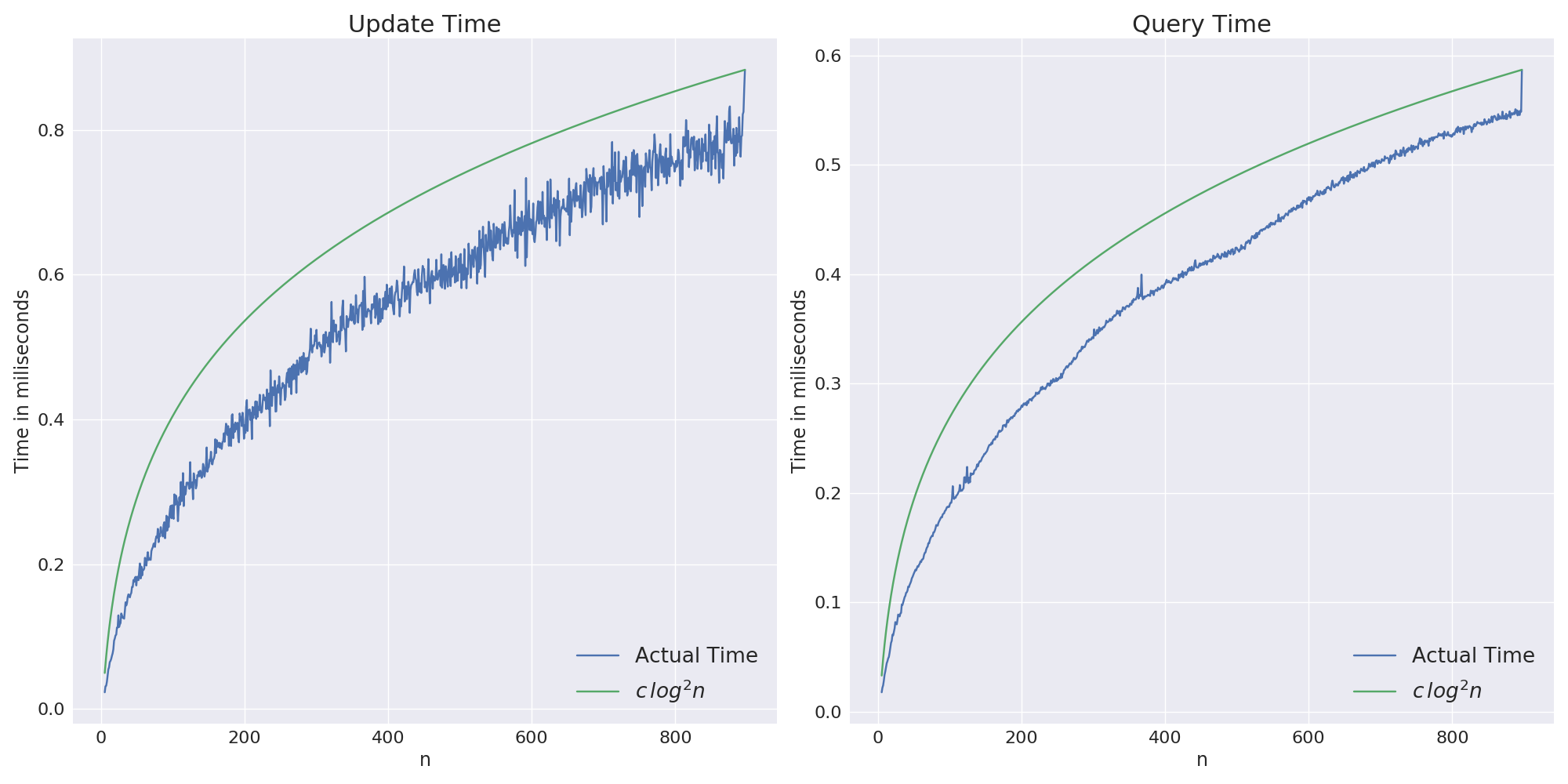}
    \caption{Average time needed for update and query operations by the proposed Segment Tree}
    \label{fig:time}
\end{figure}

Computing the time duration of the operations is not reliable enough as some resources of the CPU may be occupied by some other processes during one computation but free during some other computation. CPUs are also prone to thermal throttling. Thus, we also calculated the average number of nodes visited for the operations. These results are presented in Figure \ref{fig:step}. Similarly, in this case, it can be seen that the required number of steps follow the shape of a fitted $c\log^2 n$ curve.

\begin{figure}
    \centering
    \includegraphics[width=\textwidth]{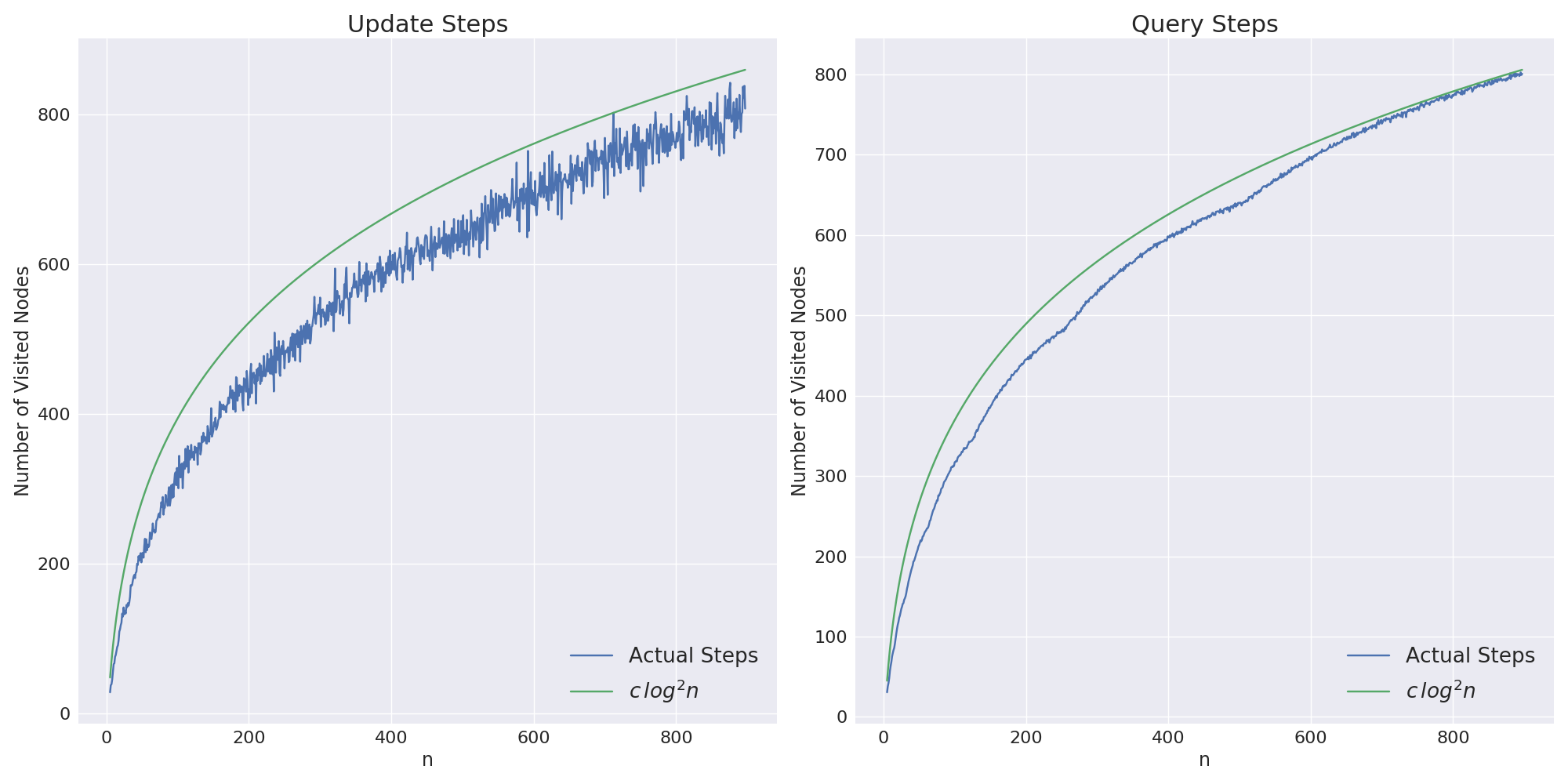}
    \caption{Average number of nodes visited during update and query operations by the proposed Segment Tree}
    \label{fig:step}
\end{figure}

\section{Discussions}
\label{sec:dis}
\subsection{Generalization}

Till now, we have focused on modifying the Segment Tree data structure in such a way that it can solve the range sum query problem over a two-dimensional array. However, the proposed algorithm can not only be extended to higher-dimensional problems but also can be utilized with other types of aggregate functions. In this section, we briefly mention the generalization of our approach. First, we present the following theorem.

\begin{theorem}
\label{thm:multi}
The proposed algorithm can perform range sum query on $d$-dimensional data in $O(\log^d n)$ time.
\end{theorem}
\begin{proof}
For $d=1$, we can use the classical 1D Segment Tree, thus it suffices to prove this theorem for $d\geq2$. We prove this theorem using mathematical induction.

\textbf{Base Case}: $k=2$

Base case is proved by Theorem \ref{thm:corr} and Section \ref{sec:time}.

\textbf{Induction step}: Let, for $p \in \mathbb{Z}$, $p > 2$, the statement is true.

Hence, we are capable of performing dynamic range sum queries on data of $p$-dimensions in $O(\log^p n)$ time. This also establishes that for $p$-dimensional data it is possible to implement the proposed scheme of lazy propagation by using the `Dispersed Updates' and `Partial Queries'. Now, let us introduce another external dimension to the $p$-dimensional problem. Since the $p$-dimensional subspace already supports lazy propagation within it, we can simply propagate the global values from the outermost new dimension to the inner $p$-dimensional space. Thus, lazy propagation can be implemented. Furthermore, assuming the range of the outer dimension is of length $n$, it can be divided into two segments at most $O(\log n)$ time. Thus, the overall time complexity becomes $O(\log n \times \log^p n) = O(\log^{(p+1)} n)$. Hence, the result follows.

%
\end{proof}


Finally, we discuss how the proposed algorithm can be extended to other aggregate functions as well. In the proposed algorithm we pass the information from $x$ dimension to $y$ dimension through the use of `Partial Query' and `Dispersed Update'. Together, they mimic the lazy propagation procedure. The most important part in both these ideas is to scale the regions properly. Through scaling, we acknowledge the effects of the individual elements of the subarray. Since the operations always impose an identical change to all the elements of a region, by considering how the individual elements change, we can segment away the effect of updating a certain region of the whole space. Thus almost any aggregate function can be scaled in this away and only a selected portion of the region can be considered.

Since all our analyses were focused on performing sum queries, we performed divisions to scale them. However, if we were to perform multiplication queries it would require computing the roots to ensure proper scaling. On the other hand for operations like AND or OR, no complicated scaling is necessary. Hence it is always possible to extract how an individual element is modified while performing an aggregate operation and it is possible to split up the effect of updating a specific portion of the region. Thus, the proposed algorithm can be extended to other aggregate functions as well.

Our Python\cite{28van2007python} implementation of the proposed algorithm along with generalization to solving other dynamic range queries using this algorithm can be found in: 

\centerline{\url{https://github.com/nibtehaz/Multidimensional-Segment-Tree}}

\subsection{Comparison with other approaches}

Range query problems are frequently used in computer science. For one dimensional variant, these problems can be solved efficiently and effortlessly using Segment trees \cite{halim2013competitive}. In this section, we describe the relative suitability of the various algorithms and data structures in order to solve the two-dimensional range sum query problem and present a comparison with our proposed algorithm.

In a naive brute force manner, one can update ranges and compute the range queries of a $n \times m$ array in $O(nm)$ time, which becomes highly inconvenient when the number of queries is high. As a result, a number of data structures and algorithms have been proposed to solve the range query update problem efficiently.

Sparse tables, developed by using the ideas of dynamic programming \cite{halim2013competitive}, are capable of returning queries in constant time. However, the downfall of this data structure is that it is immutable; after initializing it in $O(nm\log n\log m)$ time, the elements in the sparse table cannot be updated. Square root decomposition, also known as Mo's algorithm, partitions the 2D array into $p \times q$ grids with $p=\sqrt{n}, q=\sqrt{m}$ \cite{halim2013competitive}. This data structure precomputes the sums of the elements of these grids. Using these grid information, this data structure is capable of updating the entries, but still it requires a time complexity of $O(\sqrt{n} \times \sqrt{m})$.

Tree-based data structures are promising in this regard. Quadtrees follow a divide and conquer approach, by recursively dividing the entire region into four squares or near square subregions, and combining them in order to perform range queries \cite{samet1988overview}. This procedure is quite similar to that of a Segment Tree. Although it performs query and update operations in $O(\log n)$ (provided $n=m$) time when working on a square region, it suffers greatly when the regions are not of the square shape. In the worst cases, for example, when the regions are of size $1 \times n$ or $n \times 1$, it has to traverse all the way down to the individual leaf nodes. This makes the time complexity $O(n)$ (for $d$-dimensional problems this becomes $O(n^{d-1})$) \cite{mishra2013new}. Fenwick Trees were designed to work on cumulative frequency tables \cite{fenwick1994new}; however, they can be adapted to perform range queries as well. The original Fenwick Tree was only capable of performing point updates and range query operations. But they can be slightly modified to perform range update and point query as well \cite{halim2013competitive}. Mishra developed a novel approach to perform range update and range query simultaneously using  Fenwick Trees \cite{mishra2013new}. Although this data structure manages to perform these operations in $O(\log n \log m)$ time ($O(\log^2 n)$, when $n\geq m$) and $O(n \times m)$ memory, it requires an exponential number of trees to do so ($4^d$ for $d$ dimensional problem) \cite{mishra2013new}. Furthermore, in order to perform a query it requires 4 query operations on the trees, and to update it demands a total of 36 update operations on the trees, each running in $O(\log n \log m)$ time. These requirements increase exponentially with the number of dimensions. All these issues make this data structure quite cumbersome, despite that it is the only data structure capable of performing these queries in asymptotically poly-logarithmic time \cite{mishra2013new}. Moreover, the algorithm presented for this data structure is strictly for range sum query only, without any provision of generality.

In this paper, we have presented an innovative approach to perform range update and range queries using 2D Segment Trees. Our proposed algorithm requires $O(\log n \log m)$ time for these operations ($O(\log^2 n)$, when $n\geq m$) and $O(n \times m)$ memory. Also, our algorithm generalizes for higher-dimensional cases and other aggregate functions similar to the original Segment Tree.




\section{Conclusion}
\label{sec:conc}

In this paper, we have developed a novel approach to perform dynamic range queries and updates efficiently on higher dimensional data using  Segment Trees. This introduces a lot of new opportunities to utilize this highly versatile data structure in solving complex multi-dimensional problems. The future directions of our research will be to reduce some unorthodox problems to range query and update problems and exploit the effectiveness of the proposed Segment Tree in solving such problems.

\appendix

\section{Overview of the Original Segment Tree Algorithm}
\label{sec:aover}

\subsection{Definition}
A Segment Tree $T$, defined on an array $A$, is a complete binary tree \cite{de1997computational}. Any intermediate node $node_i$ operates on a segment $i_{start}\;to\;i_{end}$ and stores the value of a function $f([A[i_{start}], \dots ,A[i_{end}]])$ computed on that segment. The node $node_i$ has left and right child nodes $node_l$ and $node_r$ , who are defined on the ranges $[i_{start} \dots  \floor {\frac{i_{start}+i_{end}}{2}}]$ and $[\floor{\frac{i_{start}+i_{end}}{2}} + 1 \dots i_{end} ]$ respectively. This formulation goes on until the leaf nodes of the tree are reached. A leaf node $node_{leaf}$ is coupled to only a particular element $A[leaf]$ of the array $A$, and it stores the value of the function $f([A[leaf]])$ computed on that specific element.

Segment Tree data structures are particularly useful in updating and querying on a segment, as apparent from their name. Although Segment Trees are capable of computing a diverse set of functions, the simplest application of Segment Tree is to solve the range sum query problem.

\subsection{Updating the Segment Tree}
\label{sec:update1d}

Segment Tree can modify an element of the array $A$, it is defined on and subsequently, update the values of all the corresponding segments efficiently in logarithmic time. Let, we want to add a constant value $c$ to the $i^{th}$ element of the array $A$. In order to update that element we first start from the root node of the Segment Tree, which covers the entire segment $1 \; to \; n$. Next, we divide the segment into two segments $1 \; to \; \floor{\frac{1+n}{2}}$ and $\floor{\frac{1+n}{2}+1} \; to \; n$, which corresponds to the left and right child of the root node respectively. Then, we visit the node that represents the segment containing the index $i$, and repeat the process until the actual leaf node operating on index $i$ is reached. Hence, we update the value of that node by adding $c$ to it. Finally, we backtrack to the root node and along the way, update all nodes by setting their value to the sum of their two children.

\begin{equation}
leaf\_node.value = leaf\_node.value + c    
\end{equation}

\begin{equation}
intermediate\_node.value = left\_child.value + right\_child.value
\end{equation}
\subsection{Querying on a Range}

In order to query the sum of the elements of a range $[i, \dots ,j]$, we follow an approach similar to the update operation. Again we start from the root node and start dividing the segments into two equal parts until we obtain segments that completely lie within the query range. Then we recursively add the values of all such segments and return their sum. The query operation is also performed in logarithmic time $O(logn)$

\subsection{Lazy Propagation}
\label{app:lazy}

In Section \ref{sec:update1d} we described the algorithm to update a specific array element using Segment Tree. In order to update a range, we can simply repeat the step $m$ times, where, $m$ is the length of the range. However, this leads to a time complexity of $O(nlogn)$, considering the worst case scenario, i.e., updating the entire array. Segment Tree overcomes this limitation by following an elegant process called lazy propagation. In order to do so, we keep another value called `lazy update' in each of the nodes. The lazy update value accounts for updating all nodes that are predecessors of the current node, or equivalently all proper subsegments of the segment we are working on. During querying the sequence these lazy update values are passed on to the child nodes and thus the lazy updates are propagated to the deeper nodes of the tree.

In this process, instead of strictly going to the individual leaf nodes, if in any state we are at a node where the corresponding segment completely lies within the update range, we perform a lazy update at that node. We do this by updating the value of the node by adding the product of $c$ and the length of the segment covered by that node. We also increment the lazy update value by $c$, which signifies that the elements of all the subsegments of this segment should also be updated by adding $c$ to them. This method of updating quite resembles the query operation, where not only we work on leaf nodes but also work on intermediate nodes enclosed in the query segment. 

Lazy propagation allows the Segment Tree to update ranges efficiently in logarithmic time $O(logn)$ as well.

\subsection{Higher Dimensional Segment Trees}

The 1 Dimensional Segment Tree can be extended to higher-dimensional cases as well. The most basic one is the 2 Dimensional Segment Tree which is defined on a matrix or 2D Array $A$ of size $n \times m$. 

In this case, we have a Segment Tree $T$, which is a complete binary tree. The nodes of this tree operate on segments along the first dimension, i.e., the root node is defined on the range $1 \; to \; n$. Similar to the one dimensional Segment Tree, the root node has a left child and a right child who operates on ranges $1 \; to \; \floor{\frac{n+1}{2}}$ and $\floor{\frac{n+1}{2}}+1 \; to \; n$ respectively. The child nodes of these intermediate nodes similarly divide the working segment into two subsegments, and this division goes on until the leaf nodes are reached.

However, for this variant of Segment Tree instead of the nodes containing the computed value of a function over a segment, they contain another Segment Tree inside. This second layer of Segment Tree $T'$, is on the contrary defined on segments along the 2nd dimension. This implies that the root node of tree $T'$ operates on the segment $1 \; to \; m$, its child nodes operates on ranges $1 \; to \; \floor{\frac{m+1}{2}}$ and $\floor{\frac{m+1}{2}}+1 \; to \; m$, and so on.

Thus, effectively a 2D Segment Tree is a Segment Tree of Segment Trees. This two-layer representation offers two ranges along two different dimensions. These two ranges define a rectangular region of the two-dimensional array the tree is defined on. For example, suppose we are at a node $node_i$ of the 1st layer Segment Tree, which works on the segment $i_1 \; to \; i_2$ along the 1st dimension. Now, inside the node $node_i$ the 2nd layer of the Segment Tree is defined. Let us assume in the 2nd layer we are at a node $node_{i_j}$, that operates on the range $j_1 \; to \; j_2$ along the 2nd dimension. Thus we are at the same time considering two ranges, $i_1 \; to \; i_2$ along the first dimension and $j_1 \; to \; j_2$ along the second dimension. This limits our analysis to the region $A[i_1:i_2,j_1:j_2]$. Thus, using the two-dimensional Segment Tree we can perform updates and queries on a subregion of a two-dimensional array.

However, the two-dimensional Segment Tree can be considered only for point updates \cite{algomax}. It falls short that it does not support lazy propagation. Mishra \cite{mishra2013new} proved that the 2D Segment Tree can perform lazy propagation only along the last i.e., the 2nd dimension only. Which effectively makes the time complexity of the update operation $O(nlogn)$.

Similarly, by cascading more layers of Segment Trees higher dimensional Segment Trees can be constructed and utilized. However, these higher dimensional variants also fall short to lazy propagation, as they are only capable of supporting it along the innermost dimension \cite{mishra2013new}. As a result, for a $k$ dimensional tree, the time complexity becomes $O(n^{k-1}logn)$

\clearpage
\bibliographystyle{unsrt}
\bibliography{bibilography}

\end{document}